\documentclass[12pt]{amsart}
\usepackage{amsthm,amsmath,amsxtra,amscd,amssymb,xypic,color}
\usepackage{graphicx,color}
\usepackage[T1]{fontenc}
\def\beq{\begin{equation}}
\def\eeq{\end{equation}}
\def\bea{\begin{eqnarray}}
\def\eea{\end{eqnarray}}
\def\p{\partial}
\def\G{\Gamma}
\def\g{\gamma}
\def\s{\sigma}

\def\zb{\bar z}

\def\a{\alpha}
\def\b{\beta}

\def\res{{\rm res}}

\def \matrix #1 {\left(\begin{array}{cc} #1 \end{array}\right)}

\newtheorem{theo}{Theorem}[section]

\newtheorem{cor}[theo]{Corrolary}
\newtheorem{lem}[theo]{Lemma}

\theoremstyle{definition}

\newtheorem{rem}[theo]{Remark}

\def\bpf{\hfill\break{  \textbf{ Proof:}} }

\numberwithin{equation}{section}

\theoremstyle{definition}

\title[Integrable 2D Schr\"odinger operators]{Two-dimensional Periodic Schr\"odinger Operators Integrable at Energy Eigenlevel}
\author{A.Ilina, I.Krichever, N.Nekrasov}
\address{Skolkovo Institute for Science and Technology; National Research University "Higher School of Economics", Moscow}
\email{ekrez@yandex.ru}

\address{Skolkovo Institute for Science and Technology, Moscow; Columbia University, New York; National Research University "Higher School of Economics", Moscow}
\email{krichev@math.columbia.edu}

\address{Skolkovo Institute for Science and Technology, Moscow; Simons Center for Geometry and Physics, Stony Brook}
\email{nikitastring@gmail.com}
\thanks{The work of A.I was supported by the Russian Academic Excellence Project 5-100 and the RFBR grant 18-01-00273a.}
\begin{document}
\maketitle
\begin{abstract}\vspace*{2pt}
\noindent
The main goal of the first part of the paper is to show  that the Fermi curve of a two-dimensional periodic Schr\"odinger operator with nonnegative  potential whose points parameterize the Bloch solutions of the Schr\"odinger equation at the zero energy level is a smooth $M$-curve. Moreover, it is shown that the poles of the Bloch solutions are located  on the fixed ovals of  an antiholomorphic involution so that each but one oval contains precisely one pole. The topological type is stable until, at some value of the deformation parameter, the zero level becomes an eigenlevel  for the Schr\"odinger operator on the space of  (anti)periodic functions. The second part of the paper is devoted to the construction of such operators with the help of a generalization of the Novikov--Veselov construction.
\end{abstract}

\section{Introduction}

The theory of periodic two-dimensional operators integrable at  {\it one energy level} goes back to the work
\cite{dkn}, in which an algebraic-geometric construction of integrable two-dimensional Schr\"odineger operators
\beq \label{magnetic}
\tilde H=(i\p_x+A_1(x,y))^2+(i\p_y +A_2(x,y))^2+u(x,y)
\eeq
in magnetic field was proposed. The shift  $u\to u-E$ of the potential transforms the equation $\tilde H\psi=E\psi$ into $\tilde H\psi=0$. Hence, without loss of generality, it will  always be assumed that the level equals  zero.

The construction of the work \cite{dkn} is based on the notion of the {\it two-point two-parameter  Baker-Akhiezer function} $\psi(x,y,p)$, which is uniquely determined by a smooth genus $g$ algebraic curve $\G$ with two marked points $P_\pm$ and an effective nonspecial divisor  $D=\g_1+\cdots+\g_g$. The Baker-Akhiezer function and the coefficients of the operator $\tilde H$ were explicitly written in terms of the Riemann theta-function associated with the curve  $\G$.

In \cite{nv1} and \cite{nv2} Novikov and Veselov found  sufficient conditions on the algebraic-geometric data $\{\G,P_\pm,D\}$  under which corresponding  operators
\beq
H=-\Delta+u(x,y), \ \Delta:=\p_x^2+\p_y^2=4\p_z\p_{\bar z}
\eeq
are potential, i.e.,  $A_i\equiv 0$. The corresponding curves must have a  holomorphic involution $\s:\G\to \G$ with exactly  two fixed points $P_\pm=\s(P_\pm)$. We emphasize that the latter condition turns out to be crucial for another remarkable Novikov-Veselov result, namely, an explicit expression for  the corresponding Baker-Akhiezer functions in terms of the Prym theta-function.

The nature of the Novikov-Veselov conditions was clarified in the work \cite{kr-spec}, where a constructive description of the {\it complex Fermi curve} of a two-dimensional Schr\"odinger operator was proposed. By definition,  the complex Fermi curve is a Riemann surface $\G^{{\rm Fermi}}$ (in what follows, we use the notation  $\G^{\rm F}$) whose points $p\in \G^{{\rm F}}$ parameterize the solutions of the equation
\beq \label{shrod}
\left(-\Delta+u(x,y)\right)\psi(x,y,p)=0,\ \ \Delta:=\p_x^2+\p_y^2=4\p_z\p_{\bar z},
\eeq
which are eigenvectors for the monodromy operators
\beq\label{bloch1}
\psi(x+2\pi \ell_1,y,p)= w_1(p)\psi(x,y,p), \ \ \psi(x,y+2\pi \ell_2,p)=w_2(p)\psi(x,y,p).
\eeq
In the spectral theory of periodic operators these solutions are called Bloch-Floque solutions, and the set of pairs
$(w_1,w_2)\in (C^*)^2$ for which there exists a  Bloch solution with such Floque multipliers is called the Bloch-Floquet locus. Below it will be denoted by $\G^{ BF}$. The statement that this locus is defined by an analytic equation $R(w_1,w_2)=0$, i.e., is an analytic complex curve, was proved by Taimanov as a corollary of Keldysh's theorem on the resolvent of a family of completely continuous operators (for details and the history of the question, see  \cite{kr-spec} and \cite{taim}). The Fermi curve is a partial normalization of the Bloch-Floque curve, i.e., there is a holomorphic map $\G^{ F}\longmapsto \G^{ BF}$ which is one-to-one outside the  singular points of $\G^{ BF}$ (and their preimages on  $\G^F$). For a generic potential, the curve $\G^{BF}$ is smooth and the notions of the Fermi curve and the Bloch-Floque spectral curve coincide.
\smallskip

One of the goals of this work is to generalize the Novikov-Veselov construction for the case of {\it zero eigenlevel}, namely, for the case where $E=0$ is an eigenvalue of the operator $H$ in the space of (anti)periodic functions. This case is special in many respects. In particular,  if the multiplicity of this level is odd, then the corresponding  Fermi curve  is necessarily singular.

Let us clarify the last statement. The Schr\"odinger operator is self-adjoint. Therefore, if $\psi$ is a Bloch solution of (\ref{shrod}) with the multipliers $(w_1, w_2)$, then there is a dual Bloch solution $\psi^{\s}$ of the same equation with the multipliers $(w_1^{-1}, w_2^{-1})$, i.e., the Fermi curve is invariant under the holomorphic involution
\beq\label{inv}
\sigma:  \G^{f} \longmapsto \G^{f}, \ (w_1, w_2)\longmapsto  (w_1^{-1}, w_2^{-1}).
\eeq
This involution has fixed points  only when $E=0$ is an eigenlevel of the  (anti)periodic problem for the operator $H$. These points are the nodal points of the Fermi curve, which is therefore  singular.

It is known that any singular curve admits a normalization, i.e., a holomorphic map $\nu:\Gamma\to \G^{f}$ of a smooth curve $\G$ to $\G^{f}$ which is one-to-one outside  the singular points of $\G^{f}$ (and its preimages on $\G$). Although the fact that the Fermi curve is singular is important for the motivation of the construction, the construction  itself is based on a description of the analytical properties of the preimage of the Baker-Akhiezer function on the normalization of the spectral curve, i.e., on the  smooth curve $\G$. The corresponding curves have a  holomorphic involution with $n+1$ pairs of fixed points $(P_\pm, p^1_\pm,\ldots p^n_\pm)$ . Note that the analytical properties of the Baker-Akhiezer functions  on such curves for $n>0$ coincide with those of functions leading to the so-called "multisoliton on the finite-gap background" solutions of the integrable equations. In terms of the Riemann theta-function these  Baker-Akhiezer functions are expressed  by complicated determinant-type formulas. It has turned out that in the case under consideration their expression in terms of  appropriately defined Prym theta-functions remains the same for any $n$.

The generalized Novikov-Veselov construction is presented in Section 4. From  results of  Sections 2 and 3 it will be clear that it plays a special role in the spectral theory of periodic Shr\"odenger operators. We emphasize that, despite the significant progress made in \cite{kr-spec}, the problem of constructing of the periodic spectral theory of such operators in full generality  remains  open. In particular, there is no proof that the  sufficient conditions   on the algebraic-geometric spectral data found in \cite{nv1},\cite{ kr-spec} and \cite{ nat}  are  also necessary for  Schr\"odinger operators with  real nonsingular potential. These conditions depend on the topological type of two natural commuting antiholomorphic involutions of $\G^f$. The first one is
\beq\label{invtau}
\tau:  \G^{f}\longmapsto   \G^{f}, \ (w_1, w_2)\longmapsto (\bar w_1, \bar w_2);
\eeq
it reflects the fact that the potential is real. Indeed, if  $\psi$ is a Bloch solution of a Schr\"odinger equation with  real potential, then  $\bar \psi$ is also a Bloch solution of the same equation. Since the points of the Fermi curve parameterize all  Bloch solutions, it follows that, for any $p\in\G^f,$ there is a point $\tau(p)$ such that $\bar \psi(x,y,p)=\psi(x,y,\tau(p))$. The points of the fixed ovals of $\tau$ parameterize the  real Bloch solutions of the Schr\"odinger equation.

The fixed points of the second antiholomorphic involution $\s\tau$ correspond to the pairs of the multipliers such that $|w_i|=1$. Note that
\beq\label{stau}
\s\tau(p)=p \ \ \Rightarrow \ \psi(x,y,\s(p))= \bar \psi(x,y,p).
\eeq
The simple observation, that anti-involution $\s\tau$  of the Fermi-curve corresponding to a Schr\"odinger operator with  nonnegative potential has no fixed points has turned out to be crucial for the proof that such a curve is an  $M$-curve with respect to the anti-involution $\tau$ and that the points of the pole divisor of the Bloch solution, which, together with the curve, uniquely determine the potential, are located  on  the fixed ovals of $\tau$ so that each but one oval contains precisely one such point. A similar description of the spectral data is well known in the theory of the one-dimensional Schr\"odinger operator \cite{dnm}.

For any real potential, the anti-involution $\sigma\tau$ has  only  finitely many  fixed ovals. Under a continuous deformation of the potential this number  changes only if  the curve $\G^f$ becomes singular. More precisely, it changes if  the zero energy level becomes an eigenlevel at some value of the deformation parameter. This explains the special role of the generalized Novikov-Veselov potentials.

\section{Preliminaries}

\medskip
To begin with, we present a description of the spectral Fermi curve and the Bloch-Floque curve for the simplest case of the "free" Schr\"odinger operator $H_0=-\Delta$ at the level  $E\neq 0$. The Bloch solutions of Eq. (\ref{shrod}) with $u=-E$ are parameterized by a nonzero complex parameter $k\in C^*$ and have the form
\beq\label{psifree}
\psi(z,\bar z,k)= e^{kz-k^{-1}\frac{E}{4}\bar z};
\eeq
i.e., in this case, $\G^f$ coincides with $C^*$. The corresponding Floque multipliers are given by the formulas
\beq\label{wfree}
w_1(k)=e^{2\pi(k-k^{-1}\frac{E}{4})\,\ell_1},\ \ w_2(k)=e^{2\pi i(k+k^{-1}\frac{E}{4})\,\ell_2},
\eeq
which define a map
$$ W: C^*\longmapsto  (C^*)^2, \ k\longmapsto (w_1(k),w_2(k)).$$
The image of $W$ is the spectral Bloch-Floque curve $\G_0^{BF}$ of the operator $(H_0-E)$. The holomorphic and antiholomorphic involutions (\ref{inv}) and (\ref{invtau}) in this case have the form
\beq\label{inv2}
\s: k\longmapsto -k,\ \ \tau: k\longmapsto -\frac E{4\bar k}.
\eeq
The only singularities of $\G_0^{BF}$ are the images of the   "resonant" points  $k$ and $k'$  defined by the equations $w_i(k)=w_i(k')$. Such pairs are parameterized by pairs of integers $n, m$  and the sign. They are solutions of the equations
\beq\label{res}
k-\frac{E}{4k}-\left(k'-\frac {E} {4k'}\right)=\frac{ i n} {\ell_1},\ k+\frac E {4k}-\left(k'+\frac E {4k'}\right)=\frac{ m} {\ell_2}.
\eeq
Solving (\ref{res}) we obtain  $k=k_{n,m}^{\pm}$ and $k'=k_{-n,-m}^{\mp},$ where
\beq\label{k-nm}
k_{nm}^{\pm}:=\frac{m \ell_1+i n \ell_2}{4\ell_1\ell_2}\left(1\pm \sqrt {1-\frac{4E\ell_1^2\ell_2^2}{m^2l_1^2+n^2l_2^2}}\,\right).
\eeq
The coordinates of the nodal points of the curve $\G_0^{BF}$ equal
\bea\label{w-nm}
w_1(k_{n,m}^{\pm})=\exp\left(\pi in\pm \frac {\pi m \ell_1} {\ell_2}\sqrt {1-\frac{4E\ell_1^2\ell_2^2}{m^2\ell_1^2+n^2\ell_2^2}}\,\right),\\
w_2(k_{n,m}^{\pm})=\exp\left(\pi im\mp\frac {\pi n\ell_2} {\ell_1}\sqrt {1-\frac{4E\ell_1^2\ell_2^2}{m^2\ell_1^2+n^2\ell_2^2}}\right).
\eea
The values of $E$ for which the  radicand in these  formulas vanishes, i.e.,
\beq\label{E-nm}
E_{|n|,|m|}=\frac{m^2\ell_1^2+n^2\ell_2^2}{4\ell_1^2\ell_2^2},
\eeq
are eigenvalues of the operator $H_0=-\Delta$ on the space of (anti)periodic functions. For such $E,$ the corresponding Bloch-Floque curve has,in addition to an  infinite number of simple nodal points, one self-intersection point of multiplicity $4$. The proof of the existence of the Fermi curve for any periodic Schr\"odinger operator proposed in  \cite{kr-spec} is based on a construction and  convergence analysis of  series defining {\it formal Bloch solutions}.

For any $k_0\in C^*,$ we introduce complex numbers $k_\nu$ as solutions of the equation $w_1(k)=w_{10}:=w_1(k_0),$ where $w_1(k)$ is defined in (\ref{wfree}). The indices  $\nu$ of  $k_\nu$ are pairs  $(n,\pm)$ of an integer and a sign. It is easy to see that
\beq\label{knu}
k_\nu=\frac { in} {2\ell_1}+\frac 12\left(k_0-\frac E {4k_0}\right)\pm \frac{1}{2}\sqrt {\left(\frac{ in} {\ell_1}+
 \left(k_0-\frac E {4k_0}\right)\right)^2+E}.
\eeq
 We set  $\psi_\nu:=\psi(x,y,k_\nu)$ and $w_{2\nu}:=w_2(k_\nu),$ where $\psi$ and $w_2$ are given by formulae (\ref{psifree}) и (\ref{wfree}), respectively.

Under the assumption $w_{20}\ne w_{2\nu}$  for $ \nu\ne 0$ we define formal series
\beq \label{F1}
F(y,k_0)=\sum_{s=1}^\infty F_s(y,k_0),
\eeq
\beq \label{varphi_sum}
\Psi(x,y,k_0)=\sum_{s=0}^\infty \varphi_s,\ \ \ \ \varphi_s=\sum_\nu c_\nu^s(y)\psi_\nu(x,y)
\eeq
by the recurrent formulas
\beq\label{formula51}
F_s=r_0^{-1}\langle \psi_0^+ v\varphi_{s-1}\rangle_x,
\eeq
\beq\label{formula52}
c_0^0=1, c_0^s=-r_0^{-1}\sum_{i=1}^s F_i\langle \psi_0^+\varphi_{s-i}\rangle_x, \ s\ge 1.
\eeq
For $\nu\ne 0,$ we set $ c_\nu^0=0,$ and for $s\ge 1,$ we set
\beq\label{formula6}
c_\nu^s=\frac{w_{2\nu}}{r_\nu (w_{2\nu}-w_{20})}\int_{y}^{y+2\pi \ell_2} \langle\psi_\nu^+\Bigl(-v(x,y)\varphi_{s-1}+\sum_{i=1}^s  2F_i\varphi_{s-i,y}\Bigr)\rangle_x+
\eeq
$$ +\frac{w_{2\nu}}{r_\nu (w_{2\nu}-w_{20})}\int_{y}^{y+2\pi \ell_2} \langle\psi_\nu^+\Bigl(F_{iy}\varphi_{s-i}+\sum_{l=1}^{s-i}F_iF_l\varphi_{s-i-l}\bigr)\Bigr)\rangle_x,
$$
where
\beq\label{rr}
\psi_\nu^+=\psi(x,y,-k_\nu), \ \ r_\nu=r(k_\nu),\ \ r(k):=4\pi i \ell_1\left(k+\frac E{4k}\right),
\eeq
Here and below $\langle f\rangle_x$ denotes the  average over the period of a $2\pi \ell_1$-periodic function $f$ of the variable $x$.
\begin{lem} The formula
\beq
\tilde\psi(x,y,k_0)=\exp \left(\int_0^y F(y',k_0)dy'\right) \Psi(x,y,k_0)\Psi^{-1}(0,0,k_0)
\eeq
defines a formal Bloch solution of the equation
\bea \label{var_eq2}
(-\partial_x^2-\partial_y^2+v(x,y)-E)\tilde \psi=0,
\eea
i.e., a solution of (\ref{var_eq2}) with  the monodromy properties
\beq
\begin{cases}
\tilde \psi(x+2\pi \ell_1,y,k_0)=w_{10}\tilde \psi(x,y,k_0),\\
\tilde \psi(x,y+2\pi \ell_2,k_0)=\tilde w_{20}\tilde \psi(x,y,k_0),
\end{cases}
\eeq
where
\beq
\tilde w_{20}=w_{20}\exp \left(\int_0^{2\pi\ell_2} F(y',k_0)dy'\right).
\eeq
\end{lem}
To describe the structure of the Fermi curve, let us fix a positive number $h>0$ and define  neighborhoods $R_{nm}^\pm$ of the resonant points (\ref{k-nm}) so that, for any $k_0\notin R_{nm}^\pm,$ the inequality
\beq\label{nei}
|w_{20}w_{2\nu}^{-1}-1|>h,\ \ \nu\neq 0.
\eeq
holds. Without loss of generality it can be assumed that $h$ is chosen so that the neighborhoods $R_{nm}^\pm$ do not intersect.
\begin{lem}\label{lem:2}\cite{kr-spec}
Suppose that  $v(x,y)$ admits an analytic continuation to some neighborhood of real values of the variables $x$ and $y$. Then there is a constant $N_0$ such that, for any $k_0\notin R_{nm}^\pm$ with $|k_0|+|k_0^{-1}|>N_0,$ the series given by formulas (\ref{F1})-(\ref{formula6}) converge uniformly and absolutely and define  Bloch solutions $\tilde \psi(x,y,k_0)$  of Eq. (\ref{var_eq2}). The function $\tilde \psi(x,y,k_0)$  is an analytic function of the variable  $k_0$
and does not vanish for any $x$ and $y$.
\end{lem}
In the resonance case (where the condition $w_{20}\neq w_{2\nu}$ does not hold) we construct formal Bloch solutions as follows. Let $I$  be any finite set of indices $\nu$ such that
\beq\label{condition_02}
w_{2\alpha}\ne w_{2\nu}, \ \alpha\in I, \ \nu\notin I .
\eeq
Then the matrix analogues of the previous formulas define formal {\it quasi-Bloch} solutions of Eq.  (\ref{var_eq2}), i.e., a set of solutions $ \hat\psi_\a(x,y,w_{10})$ of the Schr\"odinger equation with the  monodromy properties
\beq\label{quasi}
\begin{cases}
\hat\psi_\a(x+2\pi\ell_1,y,w_{10})=w_{10}\hat\psi_\a(x,y,w_{10}), \\
\hat\psi_\a(x,y+2\pi \ell_2,w_{10})=\sum_\b T_\a^\b(w_{10})\hat\psi_\b (x,y,w_{10}).
\end{cases}
\eeq
For brevity, we do not present here the corresponding series (see \cite{kr-spec}). For what follows,  it is sufficient to know an explicit form of the first two terms of the series for the matrix $T(w_{10})$:
\beq\label{tseries}
(T_0)_\a^\b=w_{2\b}\delta_\a^\b\,; \ \ \ (T_1)_{\a}^\b=\frac{w_{2\b}}{r_\b}\int_0^{2\pi\ell_2}\langle \psi_\b^+v\psi_\a\rangle_x\,dy'
\eeq
For $k_0\in R_{nm}^\pm$ with $|k_0|+|k_0^{-1}|>N_0,$ as the resonance set of indices we take the pairs $(\nu,\nu_0)$ such that $\nu=0$ and $k_{\nu_0}\in R_{-n,-m}^{\mp}$. As shown in \cite{kr-spec}, the formal series defining the quasi-Bloch solutions uniformly absolutely converge and define functions $\hat \psi_\a(x,y,w_{10})$ which are holomorphic in the variable $w_{10}\in W_{n,m}^{\pm}$. Here $W_{n,m}^{\pm}$ is the domain in the complex plane that is the image of $R_{n,m}^{\pm}$ under the map defined by the function $w_1(k_0)$. Note that   $W_{n,m}^{\pm}$ is simultaneously the image under the same map of the second resonance neighborhood $R_{-n,-m}^{\mp}$.

Consider the matrix $T(w_{10})$  in (\ref{quasi}). It depends holomorphically on  the variable $w_{10}\in W_{n,m}^{\pm}$. In the case under consideration, this  is a $2\times 2$ matrix, and the degree of the discriminant of its characteristic equation
\beq\label{char11}
\det \left( \tilde w_2 \cdot \mathbb{I}-T(w_{10})\right)=0
\eeq
equals 2. In other words, the characteristic equation (\ref{char11}) defines a two-sheeted covering $\G_{n,m}^{\pm}$ of the domain $W_{nm}^{\pm}$ with two branch points inside the domain. We say that a resonance point $k_{nm}^\pm$  is  marked if $\G_{n,m}^{\pm}$ is singular, i.e., the discriminant of the characteristic equation has one double zero.

\begin{lem} For nonmarked resonance points $k_{nm}^\pm,$ the Bloch function $\tilde \psi(x,y,k_0)$ extends analytically to the Riemann surface $\G_{n,m}^{\pm},$ where it has one simple pole. For marked resonance points, the Bloch function $\tilde \psi(x,y,k_0)$ extends analytically from the nonresonance domain to the neighborhoods $R_{n,m}^\pm$ and $R_{-n,-m}^{\mp}$.
\end{lem}
The extension of the Bloch functions to the "central domain" $R_0$ defined by the inequalities
\beq\label{center}
2\pi\ell_1\left|{\rm Re} \left(k_0-\frac E{4k_0}\right)\right|<r,\ \ 2\pi\ell_1\left|{\rm Im} \left(k_0-\frac E{4k_0}\right)\right|<N
\eeq
 where $N$ is an integer, is described in a similar way
 . The function $w_1(k_0)$ represents  $R_0$ as $2N$-sheeted covering of the  domain $W_0\subset C^*$ defined by the inequalities $e^{-r}<|w_{10}|<e^r$. For any $k_0$ such that $w_1(k_0)\in W_0,$ for the resonance set $I$ of indices we take the set of  those $k_\alpha$ for which   $k_\a\in R_0$.

\begin{lem} Under the assumptions of Lemma \ref{lem:2} there are constants $c_1$ and $c_2$ such that, for $r>c_1$ and $N>c_2,$ the series  defining the quasi-Bloch solutions converge uniformly and absolutely. The Bloch solutions of the Schr\"odinger equation extend from the nonresonance domain to the Riemann surface  $\tilde \G_0$ defined over $W_0$ by the characteristic equation (\ref{char11}) for the corresponding $2N\times 2N$ monodromy matrix. The extension is a meromorphic function on $\tilde \G_0$ with poles independent of $(x,y)$. The number of these poles does not exceed the number of  resonance pairs $k_{n,m}^{\pm}\in W_0$. In the generic case, the surface $\tilde \G_0$ is nonsingular and the number of poles equals the genus of $\tilde \G_0$.
\end{lem}
The results presented above make it possible  to describe the global structure of $\G^f$. This  is the surface  obtained from the complex plane of the variable  $k_0$ by "pasting in" \,$\tilde \G_0$ instead of $R_0$ "gluing" the surfaces $\G_{nm}^\pm$ instead of the domains $R_{n,m}^\pm$ for the non-marked resonance points.  If the number of  nonmarked resonance points is finite, then $\G^f$ has finite genus and can be compactified by two  points. The corresponding potentials are called algebraic-geometric (or finite-gap) at the zero energy level. As shown in  \cite{kr-spec}, the algebraic-geometric potentials are dense in the space of all periodic potentials.
\begin{rem} By "pasting in" we mean that   the boundaries of the  "excised" and "pasted in"  domains are identified by means of the function $w_1$ and the complex structure on the union of the complement to the excised domain and the pasted in one is defined by the condition that   the holomorphic functions in this structure are those that are  continuous on the gluing line.
\end{rem}

\section{The Fermi curve of the Schr\"odinger operator with  nonnegative potential}

The construction presented in the previous section allows us to consider $\G^F$ as a kind of perturbation of the Fermi curve of the operator $H=-\Delta-E$ with any chosen constant $E$. In \cite{kr-spec}, for definiteness, the constant $E$ was set to  4. Varying this parameter, we can  describe more effectively the structure of the Fermi curve $\G^F$ for the operator $H=-\Delta-E+v(x,y)$ when the periodic function $v$ is small enough, i.e., $|v(x,y)|<\varepsilon$, and the constant $E$ is not an  eigenvalue of $H_0$, i.e., $E\neq E_{|n|,|m|}$. Without loss of generality we will assume that the average of the perturbing potential over the torus equals zero: $\langle\langle v\rangle\rangle =0.$

In this case, the coefficients of the series defining formal Bloch  and quasi-Bloch solutions are bounded by the coefficients of the geometric progression with  exponent  $\varepsilon,$ which automatically guarantees the convergence of the series. Moreover, $\varepsilon$ can be chosen so that   there is no "central" resonance domain in the construction of $\G^F$, i.e.,  $\G^F$ is obtained by pasting in  only the Riemann surfaces $\G_{n,m}^{\pm}$ for nonmarked pairs of the resonant points and, as a consequence, $\G^F$ is  smooth. The description of $\G^F$ admits further effectivization when $E<0,$ which holds in the more general case of nonnegative potentials. We present this description below.

Consider the function
\beq\label{px}
p(k):=\frac{\ln w_1(k)}{2\pi \ell_1}=k-\frac{E}{4k}.
\eeq
The Riemann surface of the inverse function $k(p)$  is a two-sheeted covering of the complex $p$\,-plane with the two branch points
$\pm p_0=\pm \frac 1 2 {\sqrt {-E}}.$ For $E<0,$ we view this surface as being glued from two copies of the $p$-plane  cut along the real axis between the branch points $\pm p_0$. The gluing identifies  the upper (lower) edge of the cut on one sheet  with the lower (upper) edge  of the cut on the second sheet. In this realization the holomorphic involution $\s$ maps $p$ to $-p$ on the same sheet, and the antiholomorphic involution  $\tau$ maps $p$ to $ \bar p$ on the other sheet.

Let $\Pi$ be the set consisting of  a real number $p_0>0$ and pairs of complex numbers $\{p_{s}^j, j=1,2\}$, where $s$ is a finite or infinite set of pairs of integers $(n\neq 0, m)$. We call $\Pi$ {\it admissible} if
$p_{n,m}^1=-p_{-n,-m}^2,
p_{n,m}^2=-p_{-n,-m}^1$,
\beq\label{adm}
{\rm Im}\, p_{s}^{j}=\frac n{2\ell_1}, \ \left|p_{s}^{j}-p(k_{n,m}^\pm) \right|=o\left(\frac 1{n^2+m^2}\right),
\eeq
and the intervals $[p_{s}^1, p_{s}^2]$ parallel to the real axis do not intersect. Note that if the set of indices is finite, then the second set of  conditions in (\ref{adm}) is empty. For an infinite set of indices, these conditions mean that the corresponding pairs are asymptotically localized in neighborhoods of the corresponding resonance points.

For each admissible set $\Pi$ of data we construct a Riemann surface $\G(\Pi)$ from two copies of the complex $p$-plane with cuts between the points $p_0$ and $-p_0$ on both sheets, along the segments $[p_{s}^1, p_{s}^2]$ on the first sheet, and along the segments $[\bar p_{s}^1, \bar p_{s}^2]$ on the second sheet by identifying the upper (lower) edges of the cuts
between $p_0$ and $-p_0$  and along each $[p_s^1,p_s^2]$ on the first sheet with the lower (upper) edges of the cut between $p_0$ and $-p_0$ and along $[\bar p_s^1,\bar p_s^2]$ on the second sheet, and by  identifying upper(lower) side of the cut between $[p_{s}^1, p_{s}^2]$ on the first sheet with the lower (upper) side of the cut between $[\bar p_{s}^1, \bar p_{s}^2]$ on the second sheet,respectively. After gluing each of the cuts corresponds to a nontrivial cycle on the surface $\G(\Pi)$. The respective cycles will be denoted by $a_0$ and $a_s.$
\begin{theo} For any real positive periodic potential  $u(x,y)>0$ that can be analytically extended to a  neighborhood of real $x,y,$ the Bloch solutions of Eq. (\ref{shrod}) can be  parameterized by the points of the Riemann surface $\G(\Pi)$ corresponding to some admissible  data set $\Pi$. The corresponding function $\tilde \psi$ is meromorphic and has one pole on each of the cycles $a_s$.
\end{theo}
\begin{proof} For $E<0,$ the coordinates of the nodal points $w_i(k_{nm}^\pm)$ are real. As  mentioned above, for sufficiently small $\varepsilon,$ the set of resonant indices contains only two elements, which can be identified with $(n,m,\pm)$ and  $(-n,-m,\mp)$. It can be checked directly that
\beq\label{conj}
\psi(k_{n,m}^\pm)=\overline {\psi(k_{-n,-m}^\mp)},\ \psi^+(k_{n,m}^\pm)=\overline {\psi(k_{-n,-m}^\pm)}, r_{n,m}^{\pm}=\overline{r_{-n,-m}^{\mp}}
\eeq
where $r_{n,m}^{\pm}:=r(k_{n,m}^\pm)$, and the function $r(k)$ is defined in (\ref{rr}).
It follows from (\ref{conj}) and (\ref{tseries}) for $w_{10}=w_1(k_{n,m,}^\pm)$ that the monodromy matrix has the form
\beq\label{T}
T(w_1(k_{nm}^\pm))=w_2(k_{nm}^\pm)\left(\begin{array}{cc}
1    & \kappa      \\
\bar \kappa    & 1
\end{array}\right)+O(\varepsilon^2),
\eeq
where $\kappa=(r_{-n,-m}^{\mp})^{-1}\langle\langle\psi^+(k_{-n,-m}^\mp) v \psi_(k_{n,m}^{\pm}\rangle\rangle$. Note that  the diagonal elements of the monodromy matrix equal $1$ due to the assumption that the the average of $v$ equals zero.

From (\ref{T}) it follows that the eigenvalues of the matrix $T(w_1(k_{nm}^\pm))$ are real and distinct in the first order in the parameter $\varepsilon$. Therefore, since  because the Fermi curve is invariant under $\tau$, they should be real to all orders in $\varepsilon$. In other words, in a neighborhood of each nonmarked pair of the resonance points there is a "forbidden zone," that is, a fixed oval of the anti-involution $\tau$. Since the pole divisor of the Bloch function is invariant under $\tau$ and there is only one pole  in a neighborhood of the resonance pair, it follows that this pole must lie on the oval $a_s$. This proves the theorem for potentials  of the form  $u=-E+v$ for $E<0$ and sufficiently small $v$.

\smallskip

\noindent{\it Change of notation.} So far we denote the Bloch solutions of Eq. (\ref{shrod}) constructed with the help of perturbation theory by $\tilde \psi$. In the rest of the paper we will denote them by $\psi(x,y,p), p\in \G^F$. Similarly, the Floquet multipliers will be denoted by $w_i(p)$.

\smallskip
Our next goal is to prove that the established properties of the Fermi curve and the pole divisor of the Bloch function are stable with respect to a deformation of the potential  $u$ under which $u$ remains nonnegative. From the invariance of $\G^F$  with respect to $\s$ and $\tau$ it follows that the described structure can change only if $(i)$ the cycles $a_s$ for distinct $s$ touch each other (at this moment the curve $\G^F$ becomes singular) or  $(ii)$ there appears a pair  of resonance points $p,p',$ on $\G^f$  such that $w_i(p)=w_i(p')$ and these points are fixed under the involution $\s\tau$, i.e., $|w_i(p)|=1$.

The same argument as in the proof of Theorem 2.2 in \cite{kr-spec} proves that the periodicity of the potential
(reflected in the fact that on $\G^F$ the functions $w_i(p)$) are defined) is an obstruction to the merging of cycles $a_s$. An obstruction to the emergence of singularities of type $(ii)$ is the nonnegativity of $u$.

\begin{lem} The anti-involution $\s\tau$ of the Fermi curve $\G^f$ corresponding to a nonnegative potential has no fixed points.
\end{lem}
\bpf Suppose that there is a point on $\G^F$ that is fixed under the anti-involution $\s\tau(p)=p$. By definition $\psi=\psi(x,y,p)$ is a Bloch solution of Eq. (\ref{shrod}). Let us multiply the left--hand side of the equation by the dual Bloch function $\psi^\sigma$ and integrate  the resulting periodic function of $x$ and $y$ over the torus (it is periodic because the Floque multipliers of the function $\psi^\s$ are inverse to those of $\psi$). Integration by parts gives the equation
\beq\label{T2av}
\int_0^{2\pi\ell_1}\int^{2\pi \ell_2}_0 \left(\partial_x\psi\partial_x\psi^\s+\partial_y\psi\partial_y\psi^\s+u\psi\psi^\s\right)dxdy=0.
\eeq
From (\ref{stau}) it follows that $\psi^\s=\bar \psi$. Therefore, if $u\geq 0,$ then the left-hand side is strictly positive. The obtained contradiction completes the proof of the theorem.
\end{proof}
\begin{rem} Note that the simplicity of the description of the Fermi curve as the curve corresponding to some admissible data set is misleading to some extend. By construction the function $p$ is  multivalued on $\G(\Pi)$, but its multivaluedness reduces to adding an integer multiple of $i/\ell_1;$ therefore, the function $w_1(p)=\exp(2\pi \ell_1 p)$ is well defined on $\G(\Pi)$. Hence  the curves of the form $\G(\Pi)$ for any admissible set of data correspond to potentials periodic in the variable $x$. Periodicity in  $y$, which is equivalent to the existence on $\G(\Pi)$ of the second function $w_2,$ dives additional transcendental equations for admissible sets of data. Nevertheless, the description of the Fermi curves in terms of admissible data turns out to be effective enough for the study of perturbed curves.
\end{rem}
\section{The generalized Novikov-Veselov construction}

The goal of this section is to construct Schr\"odinger operators which are finite-gap (algebraic-geometric) at the zero energy level in the case where this zero level is an eigenlevel.

Let  $\G$ be a smooth algebraic curve with an involution $\sigma$  having $n+1$ pairs of fixed points
$P_\pm, p^i_\pm,\ i=1,\ldots,n$. The curve  $\G$ is a two-sheeted covering of the factor curve  $\G_0:=\G/\sigma$ branched at the fixed points. If $\G_0$ is of genus $g_0,$ then by the Riemann-Hurwitz formula $\G$ has genus $g=2g_0+n$. Below it is assumed that neighborhoods of the marked points $P_\pm$ are endowed with fixed local coordinates $k_\pm^{-1}$ which are odd with respect to the involution, i.e., $k_\pm(p)=-k_\pm(\s(p)).$

We say that a divisor $D=\g_1+\cdots+\g_{g+n}$ is admissible if the set of points $\g_s$ and $\gamma^\sigma_s$ is the set of zeros of some meromorphic differential $d\Omega$ having simple poles at the fixed points of the involution and residues satisfying the equations
\beq\label{omegares}\res_{P_{\pm}}d\Omega=\pm 1,\ \ \res_{p^i_+}d\Omega=-\res_{p^i_-}d\Omega.
\eeq
\begin{lem} \label{lemBA} For a generic  admissible divisor $D,$ there is a unique Baker-Akhiezer function  $\psi(x,y,p), p\in \G,$ such that

(i) $\psi$ is meromorphic on $\G\setminus P_\pm$ and has at most simple poles at the points $\g_s$ (if they are distinct);

(ii)  in a neighborhood of the points  $P_\pm$ the function $\psi$ has the form
\beq\label{ii}
\psi=e^{k_\pm(x\pm iy)}\left(\sum_{s=0}^\infty \xi_s^{\pm}(x,y) k_\pm^{-s}\right), \ \ k_\pm=k_{\pm}(p);\eeq

(iii) its values at the points $p_\pm^i$ satisfy the equations
\beq\label{nod}
\psi(x,y,p^i_+)=\psi(x,y,p_-^i);
\eeq

(iv) the coefficients $\xi_0^{\pm}$ in (\ref{ii}) equal
\beq\label{xio}
\xi_0^+=1, \ \ \xi_0^-=1.
\eeq
\end{lem}
\noindent
\bpf
According to \cite{kr77}, the vector space of functions that satisfy the first two conditions $(i)$ and $(ii)$ is of dimension $n+1$. Hence in the general position the $n$ linear equations in $(iii)$ and the normalization $\xi_0^+=1$ of the coefficient uniquely determine $\psi$. The proof that the second equation in (\ref{xio}) holds in the case is similar to that in \cite{nv1}.

Indeed, consider the differential $d\Omega_1=\psi\psi^\s d\Omega,$ where $d\Omega$ is a meromorphic differential whose zeros are the poles of the functions $\psi$ and $\psi^\s,$ so that $d\Omega_1$ has poles  only at the points $p^i_\pm$ outside the marked points $P_\pm.$  By assumption the local coordinates $k^{-1}_\pm$  are odd with respect to  $\s$. Therefore, the exponential singularities of  $\psi$ and $\psi^\s$ are canceled in the differential $d\Omega_1$. Hence the differential $d\Omega_1$ is meromorphic on $\Gamma$. The sum of all residues of a meromorphic differential equals 0. It follows from (\ref{omegares}), (\ref{nod})  that the sum of residues of $d\Omega_1$ at each pair of points $p_\pm^i$ equals  $0$. Hence
$$0=\res_{P_+}\psi\psi^\sigma d\Omega+\res_{P_-}\psi\psi^\sigma d\Omega=1-(\xi_0^-)^2.$$
A priory $\xi_0^-$ is a meromorphic function of its arguments. Therefore, it identically equals $1$ or $-1$.
For $x=y=0,$ the function$\psi$ equals $1$ identically in $p$. Hence $\xi_0^-(0,0)=1,$ which implies (\ref{xio}).

The analytical properties of $\psi$ are similar to those used in the construction of the so-called\, "multisoliton solutions on the finite-gap background" of the KdV equation (see \cite{kr-mult}). This makes it possible to obtain an explicit, although cumbersome,  expression for $\psi$ in terms of the determinant of an $n\times n$ matrix whose entries are expressed in terms of the Riemann theta-function. Our next goal is to show that, for any $n,$ the Baker-Akhiezer function admits a simple explicit expression in terms of appropriately defined Prym theta-functions.

It is known that on  $\G$ there is a basis of $a$- and $b$-cycles  with canonical intersection matrix: $a_i\cdot a_j=b_i\cdot b_j=0, a_i\cdot b_j=\delta_{ij};$ moreover, in this basis the action of the involution $\s$ has the form
\beq\label{sa}\sigma(a_i)=a_{i+g_0}, \ \ \sigma(b_i)=b_{i+g_0}, \ i=1,\ldots, g_0,
\eeq
and
\beq\label{sa1}\sigma(a_i)=-a_i, \ \ \sigma(b_i)=-b_i, \ i=2g_0+1, \ldots, 2g_0+n.
\eeq

Let $d\omega_i$ be a normalized basis of holomorphic differentials on  $\G$\ (i.e., $\oint_{a_j} d\omega_i=\delta_{ij}$). We introduce the following  holomorphic Prym differentials odd with respect to $\sigma$:
\beq du_i=d\omega_i-d\omega_{i+g_0},  \  i=1,\ldots, g_0,
\eeq
$$ du_i=2d\omega_{i}, i=g_0+1, \ldots, g_0+n.$$
Let $\Pi$ denote the matrix of their  $b$-periods,
$ \Pi_{i,j}=\oint_{b_i}du_j,$ which defines the corresponding Prym theta-function
\beq\label{theta}
\theta(z|\Pi):=\sum_{m\in \mathbb Z^{g+n}} e^{2\pi i(z,m)+\pi i (m,\Pi m)}.
\eeq
By $A(p)$ we denote the vector (depending on the choice of a path) with coordinates  $A_i(p)=\int_{P_+}^p du_i$ and by $\Omega_\pm(p),$ the Abelian integrals $\Omega_{\pm}=\int_{P_+}^p d\Omega_{\pm},$ where $d\Omega_+(d\Omega_-)$ is the normalized (i.e., having zero $a$-periods) meromorphic differential which has only one  singularity (at $P_+(P_-)$) and is  of the form $d\Omega_\pm =dk_\pm(1+O(k_\pm^{-2}))$.
The definition of $\Omega_+$ needs clarification, since $d\Omega_+$ has a pole at $P_+$. By the integral  $d\Omega_+$ of the point $P_+$ we mean the choice of the  branch $\Omega_+=k_++O(k_+^{-1})$ in a neighborhood of $P_+$ and the analytic continuation  along the path. In what follows, it is assumed that the paths in the definition of $A(p)$ and $\Omega_\pm(p)$ are the same.

\begin{lem} \label{lemma3.2} The Baker-Akhiezer function in Lemma \ref{lemBA} equals
\beq \label{B-A}\psi(x,y,p)=
\frac{\theta(A(p)+zU_++\bar zU_-+Z|\Pi)\theta(Z|\Pi)}{\theta(zU_++\bar z U_-+Z|\Pi)\theta(A(p)+Z|\Pi)}
e^{z\Omega_+(p)+\bar z\Omega_-(p)},
\eeq
where  $U_+$ and $U_-$ are the vectors with coordinates
\beq\label{U}
U_\pm^j=\frac{1}{2\pi i}\oint_{b_j}d\Omega_\pm
\eeq
and
\beq\label{Z}
Z=-\sum_s A(\g_s)+{\mathcal K},
\eeq
where ${\mathcal K}$ is a constant vector.
\end{lem}
\noindent
\bpf
It follows from the monodromy properties of the theta-function and the definition of the vectors $U_\pm$ that the function $\psi$ defined by the right-hand side of (\ref{B-A}) is a single-valued function of $p\in \G$. It is easy to see that this function has the required exponential singularity at the marked points $P_\pm$. Outside these points the function $\psi$ is meromorphic. Let us prove that its values at the pairs of  points $p_\pm^i,\, i=1,\ldots,n,$ coincide.

The differentials $d\Omega_\pm$ are odd with respect to the involution, i.e., $\sigma^*(d\Omega_\pm)=-d\Omega_\pm$. Therefore,
\beq\label{om}
\Omega_\pm(p_+^i)-\Omega_{\pm}(p_-^i)=\int_{p_-^i}^{p_+^i}d\Omega_\pm=\frac{1}{2}\oint_{a_{2g_0+i}}d\Omega_\pm=0.
\eeq
The Prym differentials are also odd. Therefore,
\beq\label{uuu}\int_{p_-^i}^{p_+^i}du_{j}=\frac{1}{2}\int_{a_{2g_0+i}}du_{j}=0\, \ \ j=1,\ldots,g_0,
\eeq
$$\int_{p_-^i}^{p_+^i}du_{g_0+j}=\frac{1}{2}\int_{a_{2g_0+i}}du_{g_0+j}=\delta_{ij}\in\mathbb{Z}.$$
Hence the coordinates of the vector $A(p_+^i)-A(p_-^i)$ are integers. The periodicity of  the theta-function  with respect to such vectors and relations (\ref{om}) imply  (\ref{nod}).

In a similar way one can check relations (\ref{xio}). The first of them is a direct consequence of the definition of $\psi$
by (\ref{B-A}).
To prove the second relation consider an odd cycle $a_0$ whose projection on $\G_0$ is a path connecting the points $P_\pm$. It is easy to see that it is homologous to the cycle
$$[a_0]=-\sum_{i=1}^n [a_{2g_0+i}]\in H_1(\Gamma;\mathbb{Z}).$$
This observation and (\ref{uuu}) imply that the coordinates of the vector $A(P_-)$ are integers. Hence the second relation in (\ref{xio}) holds.

The pole divisor $D=D(Z)$ of the function $\psi$ given by (\ref{B-A}) is a well-defined zero divisor of the multivalued function $\theta(A(p)+Z|\Pi)$. A standard argument proves that this divisor is of degree $g+n$. The  proof of the relation (\ref{Z}) between the vector $Z$ and the Prym--Abel transform of the divisor $D(Z)$ is also standard. To complete the proof of the lemma, it remains to prove the following assertion.
\begin{lem}\label{lem:Z} For a generic vector $Z,$ the zero divisor $D(Z)$ of $\theta(A(p)+Z|\Pi)$ is admissible, i.e., the divisor $D+D^\s$ is a zero divisor of a meromorphic differential with simple poles at the fixed points of $\s$ and residues satisfying (\ref{omegares}).
\end{lem}
We prove this statement after proving  the following theorem, which is the main result of this section.

\begin{theo} The Baker-Akhiezer function  $\psi$ given by formula (\ref{B-A}) where $Z$ is a generic vector is a solution of Eq. (\ref{shrod}) with potential

\beq\label{uu}
u(x,y)=-2\Delta\ln \theta(zU_++\bar zU_-+Z|\Pi)+E, \ \ E:=4\frac{d\Omega_-}{d(k_+^{-1})}(P_+).
\eeq
If
\beq\label{per}
2\pi\ell_1 (U_++U_-)=N^a+\Pi N^b,\ \  2\pi i\ell_2 (U_+-U_-)=M^a+\Pi M^b
\eeq
for some integer vectors  $N^a, N^b$ and $M^a, M^b,$ then  the function $u(x,y)$ is $(2\pi \ell_1, 2\pi \ell_2)$-periodic and the functions $\psi_i:=\psi(x,y,p_i^\pm)$ are eigenfunctions of the operator $H$ on the space of (anti)periodic functions.
\end{theo}
\bpf
As shown above, for any $Z,$ the function $\psi$ satisfies all conditions defining the Baker-Akhiezer function for some divisor $D(Z)$. An argument standard in the finite gap theory and based only on the uniqueness of the Baker-Akhiezer function proves (\ref{shrod}) with potential  $u=4\p_{\zb}\xi_1^+$ where $\xi_1^+$ is the coefficient in the expansion  (\ref{ii}) of the function $\psi$ at the point $P_+$. In a neighborhood of this point we have the equations
\beq\label{BR}
A(p)=-2U_+k_+^{-1}+O(k_+^{-2}),\ \ \Omega_-(p)=Ek_+^{-1}+O(k_+^{-2}),
\eeq
The first of them is a consequence of the Riemann bilinear relations. The substitution of (\ref{BR}) into (\ref{B-A}) gives (\ref{uu}). In the general case, it defines a meromorphic quasi-periodic function of the variables $(x,y)$. For comparison with the result of the previous section, consider the differentials
\bea\label{mom}
dp_1:=d\Omega_++d\Omega_- - \sum_{j=0}^{g_0+n}i\nu_j^bdu_j , \\
dp_2:=i(d\Omega_+-d\Omega_-)-\sum_{j=0}^{g_0+n} i\mu_j^bdu_j,
\eea
where the $\nu_j^b$ and $\mu_j^b$ are the coordinates of real vectors defined by the equations
\beq\label{baz}
2\pi (U_++U_-)=\nu^a+\Pi\nu^b,\  2\pi i(U_+-U_-)=\mu^a+\Pi\mu^b.
\eeq
The periods of these differentials over the basic cycles $a_j,\ b_j\in H_1(\G,\mathbb{Z})$ equal
\beq\label{permom}
\oint_{a_j} dp_1 =-{i\nu_j^b},\ \oint_{b_j} dp_1 =i\nu_j^a,
\oint_{a_j} dp_2 =-i\mu_j^b,\ \oint_{b_j} dp_2 =i\mu_j^a.
\eeq
This implies that if relations (\ref{per}) hold, then the functions
\beq\label{w}
w_j(p)=\exp\left (2\pi \ell_j\int^p dp_j\right),
\eeq
are single-valued on the curve $\G$. They are holomorphic outside the marked points $P_\pm,$ at which they have exponential singularities. Moreover, note that, by virtue of (\ref{om}), we have $w_j(p_+^i)=w_j(p_-^i)$ for all
$i=\overline {1,n}$. Therefore, the uniqueness of the Baker-Akhiezer function implies (\ref{bloch1}). Hence the
Baker-Akhiezer function is a Bloch solution of the Schr\"odinger equation (\ref{shrod}). The differentials $d\Omega_\pm$ are odd with respect to the involution $\s$. Hence $w_j(\s(p))=w_j^{-1}(p)$. The points $p_\pm^i$ are fixed under  $\s$. Therefore, $w_j^2(p_\pm^i)=1$. This completes the proof of the theorem.

\medskip
Let us now return to the proof of Lemma \ref{lem:Z}. Since the lemma is not used in the proof of the main theorem, we present only a sketch of its proof.

\medskip
As shown above, the function $\psi$ given by (\ref{B-A}) satisfies the Schr\"odinger equation. The same argument as in the proof of Lemma 2.3 in \cite{kr-spec} gives the equation
\beq\label{Differential}
d\Omega:=\frac{2idp_1}{\langle \psi_y \psi^\sigma-\psi \psi_y^\sigma\rangle_x}=\frac{-2idp_2}{\langle \psi_x \psi^\sigma-\psi \psi_x^\sigma\rangle_y},
\eeq
where, as above, $dp_1$ and $dp_2$ are the differentials given by (\ref{mom}) and $\langle\cdot\rangle_x$ and $\langle\cdot\rangle_y$ denote averaging over the variables $x$ and $y$, respectively.

\medskip
The differential $d\Omega$ is meromorphic on $\G$. Its zeros are the poles of the functions $\psi$ and $\psi^\s,$ and  it has poles at the fixed points of $\s$. From the definition of the Baker-Akhiezer function it follows that its residues at the points
$P_\pm$ equal $\pm 1$. The proof that its residues at the other branch points satisfy relations (\ref{omegares}) requires additional arguments.

Let us prove these relations  for the curves corresponding to periodic potentials. For such curves, Eq. (\ref{w}) defines a single-valued function $w_1(p)$. Let us fix a complex number $w_{10}$ and consider the points on $\G$ for which $w_1(p_\nu)=w_{10}.$ We set  $\psi_\nu=\psi(x,y,p_\nu)$ and $\psi_\nu^+=\psi(x,y,\s(p))$. The same argument as in the proof of Lemma 2.4 in \cite{kr-spec} shows  that, that for any periodic function $f(x),$ the following series converge to the function $f(x)$:
\beq\label{four1}
f=\sum_\nu r_\nu^{-1}\langle\psi_\nu^+f\rangle_x\p_y\psi_\nu=-\sum_\nu r_\nu^{-1}\langle\p_y\psi_\nu^+f\rangle_x \psi_\nu,
\eeq
where $r_\nu:=\langle \p_y\psi_\nu \psi_\nu^\sigma-\psi_\nu \p_y\psi_\nu^\sigma\rangle_x$. Strictly speaking, the equation holds for $w_{10}^2\neq 1,$ since otherwise some of the points $p_\nu$ are the branch points $p_\pm^i$,
at which $r_\nu=0$. The left--hand side of (\ref{four1}) does not depend on $w_{10}$.
Therefore, letting $w_{10}^2\to 1$, we see that the singular terms of the series (\ref{four1}) for $w_{10}^2=1$ cancel each other. Since $f$ is arbitrary, it follows that the cancelation takes place for each pair of marked points. This is equivalent to (\ref{omegares}) which proves Lemma \ref{lem:Z}.

\begin{rem}\label{rem} The space of all curves $\G$ corresponding to periodic Schr\"odinger operators with fixed periods is of dimension $g_0-1$. Indeed, every such curve $\G$ is defined by a factor curve $\G_0$ and a set of $2n+2$ points on it. For fixed integer vectors $N^a, M^a, N^b$ and  $M^b,$ the conditions on the periods of the differentials stated in the theorem are equivalent to a system
of $2(g_0+n)$ equations. By definition the differentials depend also on the choice of the first curve and local coordinates at the marked points $P_\pm$. A linear transformation of the two-dimensional space of these jets corresponds to a linear transformation of the two periods of the potential. Hence the total dimension is equal to $3g_0-3+2n+2-2(g_0+n)=g_0-1.$
\end{rem}
As mentioned above, in general case, the potential (\ref{uu}) of the Schr\"odinger operator is a meromorphic function of its arguments. A potential takes real values at real values of the arguments if and only if there is an antiholomorphic involution $\tau$ on $\G$ which commutes with $\s$, i.e., $\s\tau=\tau \s$ (or, equivalently, the factor curve $\G_0$ is real) and the  following conditions on the parameter determining the Baker-Akhiezer function  are satisfied:
\beq\label{antiinv_cond}
\tau(P_+)=P_-,\, \tau^*(k_+)=\bar k_-,\, \tau (p_i^++p_-^i)=(p_i^++p_-^i),\, \tau(D)=D.
\eeq
Note that the third condition is equivalent to the condition that, for each $i,$ one has either $\tau (p^i_\pm)=p_\pm^i$ or $\tau (p^i_\pm)=p_\mp^i$.

The reality of the potential corresponding to data satisfying the constraints above follows from the relation
$$\bar\psi(\tau(p))=\psi(p)\,,$$
which, in turn, follows from the uniqueness of the Baker-Akhiezer function and the fact that the analytical properties of the two functions on the left-- and right--hand sides of the relation coincide.

Below we present two types of conditions sufficient for the corresponding potentials for the Schr\"odinger equation to be regular. The first of them is a direct generalization of the constraints proposed in \cite{nv1}.

Recall that, for any antiholomorphic involution of a smooth algebraic curve of genus $g,$ the number of fixed ovals of the antiholomorphic involution is at most $g+1$. The curves for which this number equals $g+1$ are called $M$-curves.

\begin{theo} Suppose that $\G$ is an $M$-curve whose antihilomorphic involution has fixed ovals $a_0,a_1, \ldots, a_g$ and holomorphic involutions acts as in (\ref{sa}) and  (\ref{sa1}). Suppose also that  $p_\pm^i\in a_{2g_0+i}$. Let the  points $\g_s$ of an admissible divisor $D$ of degree $g+n$ be such that each of the fixed ovals $a_1\ldots, a_{2g_0}$ and each of the segments into which the ovals $a_{2g_0+i}$ are partitioned by the points $p_\pm^i$ contains precisely one of these points. Then the corresponding potential is real and nonsingular.
\end{theo}
\bpf
The proof is standard. From formula (\ref{B-A}) it follows that the poles of the potential correspond to values of $(x,y)$ at which one of the zeros of $\psi$ coincides with $P_+$. This is impossible, since, for all $(x,y),$ each of the ovals $a_1\ldots, a_{2g_0}$ and each of the segments  $a_{2g_0+i}$ contains at least one zero,  and the total number of zeros equals $g+n=2g_0+2n$. The fact that there is at least one zero on each of the ovals and segments (at whose endpoints the values of the function $\psi$ are equal) is a corollary of the fact that the total number  of  zeros and poles of a periodic function is always even and the assumption that each of the ovals and the segments  contains one pole.

\medskip
The second type of conditions sufficient for the potential to be regular is similar to that in the theory of the KP1 equation (see \cite{kr86}).

\begin{theo} Suppose that the antiholomorphic involution  $\s\tau$  is of separating type, i.e., the complement to its fixed ovals $a_1,\ldots, a_k$ is a disjoint union of two domains $\G^\pm,\, \s\tau(\G^+)=\G^-$. Suppose also that  $p_\pm^i\in \G^\pm,$ the differential $d\Omega$ defining an admissible divisor $D$ is positive on the ovals $a_s$ with respect to the orientation induced from the domain $\G^+,$  and  $c_i:=\res_{p_i^+}d\Omega <0.$ Then the corresponding potential of the  Schr\"odinger operator is real and nonsingular.
\end{theo}
\bpf Consider the un normalized Baker-Akhiezer function
\beq\label{phi}
\phi(x,y,p):=\theta(zU_++\bar z U_-+Z|\Pi)\psi(x,y,p)
\eeq
It satisfies the same analytical conditions as $\psi$ except the normalization condition(\ref{xio}).  We have already proved that the first factor in the definition of $\phi$ is real. Hence the function$\phi$ satisfies the relation  $\phi(x,y,\tau(p))=\bar \phi(x,y,p).$ By definition the ovals $a_s$ are fixed under $\s\tau,$ and their union is the boundary $\G^+$.
Hence, for the differential $d\hat\Omega=\phi\phi^\s d\Omega,$ we have
\beq\label{contradiction}
\oint_{\p\G^+}d\hat\Omega-\sum_{i=1}^n\res_{p_+^i} \ d\hat\Omega=
\oint_{\p\G^+}|\phi|^2 d\Omega -\sum_{i=1}^n c_i|\phi(x,y, p_+^i)|^2>0
\eeq
for any values of $x$ and $y$. Suppose that the potential is singular at $(x_0,y_0)$. Then $\phi(x_0,y_0,P_+)=0$. Hence the differential $d\hat \Omega(x_0,y_0,p)$ has no pole at $P_+$, i.e., in the domain $\G^+$ it has poles only at the points $p_+^i.$ Therefore, the left--hand side of Eq. (\ref{contradiction}) equals zero for $x=x_0,y=y_0$. This is a contradiction.

\medskip
{\textbf{ Example.}} In conclusion, as an example, we present Schr\"odinger operators integrable at the zero eigenlevel corresponding to hyperelliptic curves. They are a particular case of the curves considered in the framework of the generalized Novikov-Veselov construction. Without loss of generality we can assume that a hyperelliptic curve $\G$ with $n+1$ pairs of branch points is given by the equation
$$Y^2=X\prod_{i=1}^n (X-p_+^i)(X-p_-^i).$$
We identify the branch points   $X=0, \infty$ with the marked points  $P_+$ and $P_-$, respectively. As a basis of  $a$-cycles we choose the preimages of the cuts between the points $p_\pm^i$ and denote the corresponding matrix of $b$-periods of the normalized holomorphic differentials on $\G$ by $B$. Then the Prym matrix introduced above equals $\Pi=2B$.

The values of $\psi$ in Lemma \ref{lemma3.2} at the points $p_\pm^j$ equal
\beq\label{2}
 \psi_j(x,y):=
\frac{\theta(B_j+zU_++\bar z U_-+Z|2B)\theta(Z|2B)}{\theta(zU_++\bar zU_-+Z|2B)\theta(B_j+Z|2B)}
e^{zU_+^j+\bar z U_-^j}.
\eeq
\begin{cor} The functions $\psi_i$ given by formula (\ref{2}) are solutions of the Schr\"odineger equation (\ref{shrod}) with potential
$$u(z,\bar z)=-2\Delta\ln \theta(zU_++\bar zU_-+Z|2B)+E, \ E=4\frac{d\Omega_-}{d(k_+^{-1})}(P_+).$$
\end{cor}
The two types of sufficient conditions  for the potential to be real and regular that were presented above correspond to two types of real hyperelliptic curves. The first one corresponds to real branch points $p_\pm^i=\bar p_\pm^i$ and the second one, to the case $p_\pm^i=\bar p_{\mp}^i$.

\end{document}